\pdfoutput=1
\documentclass[letterpaper]{article}
\usepackage[totalwidth=480pt, totalheight=680pt]{geometry}

\usepackage{algorithm}
\usepackage{algorithmicx}
\usepackage{algpseudocode}
\usepackage{subcaption}
\usepackage{enumerate}
\usepackage{bm}
\usepackage{amsmath}
\usepackage{amssymb}
\usepackage{amsthm}
\usepackage{pifont}
\usepackage{booktabs}
\usepackage{xcolor}
\definecolor{darkgreen}{rgb}{0,0.5,0}
\usepackage{hyperref}
\hypersetup{
    unicode=false,          
    colorlinks=true,        
    linkcolor=red,          
    citecolor=darkgreen,    
    filecolor=magenta,      
    urlcolor=blue           
}
\usepackage[capitalize, nameinlink]{cleveref}

\usepackage{thm-restate}
\theoremstyle{plain}
\newtheorem{goal}{Goal}
\newtheorem{theorem}{Theorem}

\theoremstyle{definition}

\theoremstyle{remark}

\usepackage{tikz}
\usetikzlibrary{calc, graphs, graphs.standard, shapes, arrows, arrows.meta, positioning, decorations.pathreplacing, decorations.markings, decorations.pathmorphing, fit, matrix, patterns, shapes.misc, tikzmark}

\newcommand{\cA}{\mathcal{A}}

\newcommand{\eps}{\varepsilon}

\newcommand{\wh}{\widehat}

\title{A short note about the learning-augmented secretary problem
}
\author{
Davin Choo
\quad
Chun Kai Ling
\\
National University of Singapore
}

\begin{document}

\maketitle

\begin{abstract}
We consider the secretary problem through the lens of learning-augmented algorithms.
As it is known that the best possible expected competitive ratio is $1/e$ in the classic setting without predictions, a natural goal is to design algorithms that are 1-consistent and $1/e$-robust.
Unfortunately, \cite{fujii2024secretary} provided hardness constructions showing that such a goal is not attainable when the candidates' true values are allowed to scale with $n$.
Here, we provide a simple and explicit alternative hardness construction showing that such a goal is not achievable even when the candidates' true values are constants that do not scale with $n$.
\end{abstract}

\section{Introduction}

The secretary problem is a classic online problem involving selecting the best candidate from a pool of $n$ applicants that arrive in an online fashion, where the arrival sequence is a random permutation.
When the $i^{th}$ candidate arrives, his/her value $v_i$ is revealed and one has to decide whether to irrevocably reject this candidate, or accept this candidate, ending the processing and obtaining a reward of $r = v_i \geq 0$.
One can measure the performance of a protocol or an algorithm based on the probability $\Pr(r = v^*)$ of successfully recruiting the true best candidate, where $v^* = \max_{i \in [n]} v_i$, or the competitive ratio $\frac{r}{v^*}$ of the obtained reward.
The former is often known as the classic secretary problem and is studied through the lens of optimal stopping while the latter is sometimes known as the value-maximization variant.\footnote{An $\alpha$-approximation for the former implies an $\alpha$-approximation for the latter, and the competitive ratio can be made arbitrarily close to the probability of success by simply having $v^*$ be arbitrarily large compared to all other values.}
In this note, we focus on the latter performance measure, and it is well-known \cite{dynkin1963optimum,lindley1961dynamic} that one can acheive a competitive ratio of $1/e$ by Ignoring the first $\ell \approx n/e$ arrivals and choose the first candidate $i > n/e$ such that $v_i \geq \max_{j \in \{1, \ldots, \ell\}} v_j$.

We consider the secretary problem through the lens of learning-augmented algorithms, where one is additionally given side-information or predictions about the problem instance.
The design and analysis of such methods are often explored under the framework of learning-augmented algorithms\footnote{The website \url{{https://algorithms-with-predictions.github.io/}} tracks recent progress in this research area.} \cite{mitzenmacher2022algorithms} where the typical performance measures are consistency and robustness which quantify the two extremes of perfect advice and arbitrarily bad advice.
\cite{antoniadis2023secretary} explored the secretary problem given a prediction $\wh{v}^*$ of $v^*$ while \cite{fujii2024secretary} studied the setting where one is given predictions $\wh{v}_1, \ldots, \wh{v}_n$ of the true candidate values $v_1, \ldots, v_n$.
Observe that the prediction model of \cite{fujii2024secretary} is richer than that of \cite{antoniadis2023secretary}; see \cite[Section 3.1]{fujii2024secretary} for a more in-depth comparision.
Under this richer predictive model, \cite{fujii2024secretary} propose an algorithm called \textsc{LearnedDynkin} which achieves a competitive ratio of $\max \left\{ 0.215, \frac{1-\eps}{1+\eps} \right\}$, where $\eps = \max_{i \in [n]} \left| 1 - \frac{\wh{v}_i}{v_i} \right|$ measures the maximum multiplicative error, i.e.\ $(1 - \eps) v_i \leq \wh{v}_i \leq (1 + \eps) v_i$ for any candidate $i \in [n]$.
This implies that \textsc{LearnedDynkin} is $1$-consistent and $0.215$-robust since $\frac{1-\eps}{1+\eps} = 1$ when $\eps = 0$.

Since the best known competitive ratio for the secretary problem is $1/e \approx 0.367$, a natural goal is to design algorithms that are simultaneously 1-consistent and $1/e$-robust.

\begin{goal}
\label{goal}
Is there a learning-augmented algorithm for the secretary problem that is 1-consistent and $1/e$-robust?
\end{goal}

Unfortunately, \cite{fujii2024secretary} also provided hard instances over which no algorithm using predictions of the form $\wh{v}_1, \ldots, \wh{v}_n$ can achieve \cref{goal}.
For any constant $c > 0$, they showed that no deterministic algorithm can achieve competitive ratios better than $\max\{1 - c \eps, 0.25\}$ and no randomized algorithm can achieve competitive ratios better than $\max\{1 - c \eps, 0.348\}$.
We observe that the candidate values in their hardness constructions scale exponentially with the number of candidates $n$, and the randomized construction involves an exponentially large linear program which they could only solve up till $n = 7$.

In this short note, we provide a simple and explicit alternative hardness construction (\cref{thm:hardness}) that holds for any algorithm (deterministic or randomized), which holds for any $n \geq 3$, and where the candidates' true values are constants do not scale with $n$.
Contrary to what the hardness construction of \cite{fujii2024secretary} may suggest, our result implies that one cannot hope to achieve \cref{goal} even with constant candidate values that do not scale with $n$.
Given the construction in \cref{thm:hardness}, extension to the $n > 3$ case is straightforward and the analysis will actually reduce to the $n = 3$ case.
We explain how to do so after proving \cref{thm:hardness}.

\begin{theorem}
\label{thm:hardness}
There is a secretary instance for $n = 3$, i.e.\ $\{X_1, X_2, X_3\}$ will arrive in a uniform random order, whereby any 1-consistent algorithm cannot be strictly better than $1/3 + o(1)$-robust.
\end{theorem}
\begin{proof}
Let \textsc{ALG} be an arbitrary 1-consistent algorithm and let $\eps > 0$, even $k \geq 20$ and $s > 1$ be constants to be determined later.
Consider the construction for $n = 3$ in \cref{tab:hardness-construction-n=3} whereby the realized sequence is guaranteed to be one of these $r = 2k-1$ rows, with the associated realization probabilities given, and the prediction given is the first row.
For $i \geq 3$, the term $s^i$ first appears on the row $2i - 4 + (i \mod 2)$.

\begin{table}[htb]
\centering
\begin{tabular}{cccccc}
\toprule
& $X_1$ & $X_2$ & $X_3$\\
\midrule
1 & $s$ & 1 & 1 & Realized with probability $\eps$ & (Prediction)\\
2 & $s$ & $s^2$ & {\color{blue}$s^3$} & Realized with probability $\frac{1-\eps}{r-1}$\\
3 & $s$ & {\color{blue}$s^3$} & $s^2$ & Realized with probability $\frac{1-\eps}{r-1}$\\
4 & $s$ & {\color{red}$s^4$} & {\color{blue}$s^3$} & Realized with probability $\frac{1-\eps}{r-1}$\\
5 & $s$ & {\color{blue}$s^3$} & {\color{red}$s^4$} & Realized with probability $\frac{1-\eps}{r-1}$\\
6 & $s$ & {\color{red}$s^4$} & {\color{orange}$s^5$} & Realized with probability $\frac{1-\eps}{r-1}$\\
7 & $s$ & {\color{orange}$s^5$} & {\color{red}$s^4$} & Realized with probability $\frac{1-\eps}{r-1}$\\
8 & $s$ & {\color{green!50!black}$s^6$} & {\color{orange}$s^5$} & Realized with probability $\frac{1-\eps}{r-1}$\\
9 & $s$ & {\color{orange}$s^5$} & {\color{green!50!black}$s^6$} & Realized with probability $\frac{1-\eps}{r-1}$\\
10 & $s$ & {\color{green!50!black}$s^6$} & {\color{purple}$s^7$} & Realized with probability $\frac{1-\eps}{r-1}$\\
11 & $s$ & {\color{purple}$s^7$} & {\color{green!50!black}$s^6$} & Realized with probability $\frac{1-\eps}{r-1}$\\
$\vdots$ & $\vdots$ & $\vdots$ & $\vdots$\\
$2k-4$ & $s$ & {\color{cyan}$s^{k}$} & ${\color{pink}s^{k-1}}$ & Realized with probability $\frac{1-\eps}{r-1}$\\
$2k-3$ & $s$ & ${\color{pink}s^{k-1}}$ & {\color{cyan}$s^{k}$} & Realized with probability $\frac{1-\eps}{r-1}$\\
$2k-2$ & $s$ & {\color{cyan}$s^{k}$} & $s^{k+1}$ & Realized with probability $\frac{1-\eps}{r-1}$\\
$2k-1$ & $s$ & $s^{k+1}$ & {\color{cyan}$s^{k}$} & Realized with probability $\frac{1-\eps}{r-1}$\\
\bottomrule
\end{tabular}
\caption{Hardness construction for $n = 3$. The values are constants since $k \geq 12$ and $s > 1$ are constants.}
\label{tab:hardness-construction-n=3}
\end{table}

There are three cases to consider, depending on whether $X_1$, $X_2$, or $X_3$ arrives first:
\begin{enumerate}
    \item With probability $1/3$, $X_1$ arrives first.

    With certainty, we observe $X_1 = s$.
    Since the \textsc{ALG} is 1-consistent, it has to accept $X_1$, resulting in an expected competitive ratio that is at most $\eps \cdot 1 + (1 - \eps) \cdot \frac{1}{s}$ since the competitive ratio is at most $\frac{1}{s}$ whenever any other row is realized.

    \item With probability $1/3$, $X_2$ arrives first.

    \begin{itemize}
        \item With probability $\eps$, we observe $X_2 = 1$.
        So, the realized sequence must be row 1 and \textsc{ALG} can achieve a competitive ratio of 1 by waiting for $X_1 = s$ to arrive.
    
        \item With probability $\frac{1-\eps}{r-1}$, we observe $X_2 = s^2$.
        So, the realized sequence must be row 3 and \textsc{ALG} can achieve a competitive ratio of 1 by waiting for $X_3 = s^3$ to arrive.
        
        \item With probability $\frac{1-\eps}{r-1}$, we observe $X_2 = s^{k+1}$.
        So, the realized sequence must be row $2k - 1$ and \textsc{ALG} can achieve a competitive ratio of 1 by accepting $X_2 = s^{k+1}$.
    
        \item With probability $1 - \eps - 2 \cdot \left( \frac{1-\eps}{r-1} \right) = (1-\eps) \cdot \left( \frac{r-3}{r-1} \right)$, we observe $X_2 = s^i$ for some $i \in \{3, \ldots, k\}$, i.e. one of the colored/non-black entries in the column $X_2$ shown in \cref{tab:hardness-construction-n=3}.
        No algorithm can distinguish between rows $2i - 4 + (i \mod 2)$ and $2i - 2 + (i \mod 2)$.
        If an algorithm $\cA$ accepts $X_2$, then it attains competitive ratio 1 when the realized row is $2i - 4 + (i \mod 2)$ and competitive ratio $1/s$ when the realized row is $2i - 2 + (i \mod 2)$, where each case occur with equal probability of $1/2$; a similar argument holds (with the roles of the row indices flipped) if $\cA$ rejects $X_2$.
        Consequently, any algorithm (including \textsc{ALG}) can achieve only a competitive ratio of at most $\frac{1}{2} \cdot 1 + \frac{1}{2} \cdot \frac{1}{s}$.
    \end{itemize}
    
    \item With probability $1/3$, $X_3$ arrives first.

    The analysis is exactly symmetrical with Case 2 with swapping the identities of $X_2$ and $X_3$, along with the row indices accordingly.
\end{enumerate}

Putting together, we see that the best attainable competitive ratio by \textsc{ALG} is
\begin{align*}
&\; \frac{1}{3} \cdot \left( \eps + (1-\eps) \cdot \frac{1}{s} \right) + 2 \cdot \frac{1}{3} \cdot \left( \eps + 2 \cdot \frac{1-\eps}{r-1} + (1-\eps) \cdot \left( \frac{r-3}{r-1} \right) \cdot \left( \frac{1}{2} + \frac{1}{2s} \right) \right)\\
= &\; \eps + \frac{2}{3} \cdot (1-\eps) \cdot \left( \frac{1}{2s} +  \frac{2}{r-1} +{\color{blue}\frac{r-3}{r-1}} \cdot \left( \frac{1}{2} + \frac{1}{2s} \right) \right) \tag{Refactoring}\\
< &\; \eps + \frac{2}{3} \cdot (1-\eps) \cdot \left( \frac{1}{2s} + \frac{2}{r-1} + \frac{1}{2} + \frac{1}{2s} \right) \tag{Since ${\color{blue}\frac{r-3}{r-1}} < 1$ and all other terms are non-negative}\\
= &\; \eps + \frac{2}{3} \cdot (1-\eps) \cdot \left( \frac{1}{2} + \frac{1}{s} + \frac{1}{k-1} \right) \tag{Since $r = 2k-1$}\\
= &\; \frac{1}{3} + \frac{2}{3} \cdot \left( \eps + (1-\eps) \cdot \left( \frac{1}{s} + \frac{1}{k-1} \right) \right)
\end{align*}

Defining $\alpha = \frac{1}{3} + \frac{2}{3} \cdot \left( \eps + (1-\eps) \cdot \left( \frac{1}{s} + \frac{1}{k-1} \right) \right)$ and $\beta = \frac{3}{2} \cdot \left( \frac{1}{e} - \frac{1}{3} \right) < 0.051819\ldots$, we see that
\[
\alpha < \frac{1}{e}
\iff
\eps + (1-\eps) \cdot \left( \frac{1}{s} + \frac{1}{k-1} \right) < \frac{3}{2} \cdot \left( \frac{1}{e} - \frac{1}{3} \right) = \beta
\iff \frac{1}{s} + \frac{1}{k-1} < \frac{\beta - \eps}{1 - \eps}
\]

Setting $\eps = \beta/2$, we see that
$
\frac{\beta - \eps}{1 - \eps}
= \frac{\beta/2}{1 - \beta/2}
< {\color{orange}0.0265987}
$.
If we further set $s \geq \left\lceil \frac{1/2}{{\color{orange}0.0265987}} \right\rceil = 19$ and $k-1 \geq \left\lceil \frac{1/2}{{\color{orange}0.0265987}} \right\rceil = 19$, the best attainable competitive ratio by \textsc{ALG} is strictly less than $1/e$.
That is, any arbitrary 1-consistent algorithm is strictly less than $1/e$-robust for $s = 19$ and $k = 20$ in \cref{tab:hardness-construction-n=3}.

In fact, we see that the upper bound competitive ratio $\alpha$ is independent of $n$, so we can effectively get the ratio to be $\frac{1}{3} + o(1)$.
For instance, to get $\frac{1}{3} + 0.01$, we can set $\eps = 0.01$, $s = k = 400$.
\end{proof}

\vspace{-10pt}
\paragraph{Extension to $n > 3$.}
In the construction of \cref{thm:hardness}, let $X_j = 1$ for all $j > 3$ throughout all rows.
Since $X_1 = s > 1$, any algorithm will ignore arrivals of $X_4, \ldots, X_n$. 
Furthermore, the posterior belief over rows does not change whenever such an arrival occurs (i.e.\ there is no information gain), implying that what matters is only the relative order of arrival between $X_1, X_2, X_3$.
Thus, the analysis reduces to the $n = 3$ case.

\vspace{-10pt}
\paragraph{Acknowledgements.}
The authors would like to thank Anupam Gupta for discussions and for giving feedback.
Additionally, DC would like to thank Eric Balkanski and Will Ma for discussions while he was attending the Workshop on Learning-Augmented Algorithms (August 19-21 2024) at TTIC Chicago.
DC thanks Kaito Fujii for pointing out a typo in one of the case analysis bullet points in the earlier arXiv version.

\bibliography{refs}
\bibliographystyle{alpha}

\end{document}